\documentstyle{llncs}

\title{Doubly Perfect Nonlinear Boolean Permutations}
\author{Laurent Poinsot}
\institute{LIPN CNRS UMR 7030, Institut Galil\'ee - Universit\'e Paris-Nord, 99, avenue Jean-Baptiste Cl\'ement, 93430 Villetaneuse, France} 
\begin{document}
\maketitle

\begin{abstract}
Due to implementation constraints the XOR operation is widely used in order to combine plaintext and 
key bit-strings in secret-key block ciphers. This choice directly induces the classical 
version of the differential attack by the use of XOR-kind differences. 
While very natural, there are many alternatives to the XOR. Each 
of them inducing a new form for its corresponding differential attack (using the appropriate notion 
of difference) and therefore block-ciphers need to use S-boxes that are resistant against these 
nonstandard differential cryptanalysis. In this contribution we study the functions that offer the best 
resistance against a differential attack based on a finite field multiplication. We also show that in 
some particular cases, there are robust permutations which offers the best resistant against {\bf{both}} 
multiplication and exponentiation based differential attacks. 
We call them {\it doubly perfect nonlinear permutations}. 
\end{abstract}

\section{Introduction}

Shannon has introduced in \cite{Sha} the notions of {\it diffusion} and 
{\it confusion} which have been mainly accepted and successfully used by 
cryptologists as guidelines 
in their work to design secret-key ciphers. These notions 
accurately set up a category of "nice" cryptographic objects namely 
the iterative block-ciphers such as the Data and Advanced Encryption Standards 
(see \cite{fip99,fip01}). Such an algorithm works as an iteration of a certain procedure called 
the round function. This functions is made in two pieces, 
a linear and a nonlinear parts, whose roles are to 
satisfy Shannon's diffusion and confusion. Diffusion refers to a sensitivity to the initial conditions: a small deviation in the input should cause a large 
change at the output. The linear part of the round-function is devoted to 
  provide a good level of diffusion. The goal of confusion is to hide the 
algebraic relations between the plaintext and the secret-key in order to make 
harder the statistical attacks. This is exactly the role assumed by the 
nonlinear part, also called {\it S-boxes}. One of the major attacks for which 
the S-boxes should be highly resistant is the {\it differential cryptanalysis} 
\cite{BS} or its "dual" counter-part the {\it linear attack} \cite{Mat}. 
The differential cryptanalysis is intrinsically related to the fashion the 
plaintexts and the round-keys are combined at each step. As to interlock 
plaintexts with keys, the XOR or 
component-wise modulo-two sum (or the addition in characteristic $2$) is usually chosen because of its 
implementation efficient nature. A block-cipher is then vulnerable to the 
differential attack if there is a nonzero XOR difference of two 
plaintexts such that the difference in output is statistically distinguishable 
from a random variable that follows a (discrete) uniform law. The S-boxes 
that offer the best resistance against such an attack are the 
{\it perfect nonlinear functions} \cite{Nyb}. As very particular 
combinatorial objects, perfect nonlinear functions do not exist in every 
configurations. For instance if one works in finite elementary Abelian 
$2$-groups, which in practice is usually the case, precisely because of 
the involutive nature of the addition, perfect nonlinear permutations can not 
exist. Since, yet in practice the plaintexts and ciphertexts have the 
same length, we can not use perfect nonlinear permutations as S-boxes. So in 
many cases block-ciphers exploit 
suboptimally differentially resistant functions, 
such as {\it almost perfect nonlinear} \cite{NK93} or even  
{\it differentially $4$-uniform} \cite{Nyb94} functions.\\
We make two simple observations. We have seen above that by nature, the XOR 
prohibits the existence of perfect nonlinear permutations. Moreover apart from  the XOR operation, the combination law of plaintexts and keys can take many 
forms. While really efficient by nature the XOR is a very specific case of 
group action and it could be interesting to use another one. Roughly speaking (more details are given in 
subsection~\ref{grp-action}) a group action is nothing but a particular 
external operation of a group on a set (as the scalar multiplication of 
vectors). The set in question is the collection of all the possible plaintexts. 
The set of (round) keys is endowed with a group structure and operates on the 
messages. Such a very general block-cipher could be vulnerable to a modified 
differential attack which should be no more related to the XOR differences 
but to the appropriate group action differences. In \cite{Poi06} 
is presented the algorithm of a such an attack.  Therefore the determination 
of the best resistant S-boxes or in other terms 
the adapted concept of perfect nonlinear 
functions, is needed. The theoretic description of such functions covers 
the following contributions \cite{PH04,PH05,Poi05} and the most important 
definitions and relevant results upon them are recalled in 
section~\ref{classic-general}. \\
We earlier say that altough natural, the 
XOR is not the only way to combine bit-strings. In the finite field setting 
the multiplication also may be used. The S-boxes that maximally 
resist against a 
differential attack based on the multiplication rather than the addition are 
called {\it multiplicatively perfect nonlinear functions} and in this paper 
we prove the existence of permutations with such a cryptographic property 
in many situations (and in most cases than classical perfect nonlinear 
functions). 
In addition, in some very particular cases, the 
multiplicative group ${\bbbk}^*$ of a finite field 
${\bbbk}$ in characteristic two can be equipped 
with another multiplication, which is distributive on the classical one. With 
this second multiplication (which is merely an exponentiation), 
${\bbbk}^*$ turns to be a finite field itself 
(but no more of characteristic two). 
This paper has as its major goal the construction of Boolean permutations 
over ${\bbbk}$ 
which are perfect nonlinear with respect to {\bf{both}} multiplications of the 
new field. 
They are called {\it doubly perfect nonlinear Boolean permutations} and can be seen as relevant alternatives to the use of almost perfect nonlinear 
permutations.

\section{Classical and generalized situations}\label{classic-general}

\subsection{Notations and conventions}

In this contribution the term {\it function} has the same meaning as the expression {\it total function}. 
If $X$ is a finite set then $|X|$ is its cardinality and $\mathit{Id}_X$ its 
identity map. For $f : X \rightarrow Y$ and $y \in Y$ we define as usually the fibre 
$f^{-1}(\{y\})=\{x \in X | f(x)=y\}$. 
For a additive group $(G,+,0)$ (resp. a multiplicative group 
$(G,.,1)$) we define $G^* = G \setminus \{0\}$ 
(resp. $G^*=G\setminus\{1\}$). For a unitary ring $(R,+,0,.,1)$ we have 
$R^* = R \setminus\{0\}$ and $R^{**}=
R^*\setminus\{1\}=R\setminus\{0,1\}$. Moreover the group of units of $R$ 
({\it i.e.} the group of invertible elements of the ring) is denoted $U(R)$ 
and obviously $U(R)^* = U(R)\setminus\{1\}$. In order to simplify the 
notations we sometimes identify a group (or a ring) 
with its underlying set. The ring of integers modulo $n$ is denoted 
$({\bbbz}_n,+,0,.,1)$ and its underlying set is identified with the 
particular system of representatives of residue classes 
$\{0,1,\ldots,n-1\}$. The finite field of characterisitic $p$ 
with $p^m$ elements is denoted 
${\mathsf{GF}}(p^m)$. A prime field ${\mathsf{GF}}(p)$ is identified with 
${\bbbz}_p$ and therefore with $\{0,1,\ldots,p-1\}$. 
Finally $\mathit{Aut}(G)$ denotes the set of all 
group automorphisms of a group $G$.

\subsection{Group actions}\label{grp-action}
Essential to everything that we shall discuss in this paper is the notion of 
group actions.\\
Let $G$ be a group and $X$ a nonempty set. We say that $G$ {\it acts on} $X$ if there is a group homomorphism 
$\phi:G \rightarrow S(X)$, where $S(X)$ is the group of permutations over $X$. Usually for $(g,x)\in G\times X$, we use the 
following convenient notation
\begin{equation}
g.x := \phi(g)(x)
\end{equation}
and so we hide any explicit reference to the morphism $\phi$. An action is called {\it faithful} if the corresponding homomorphism 
$\phi$ is one-to-one. It is called {\it regular} if for each $(x,y)\in X^2$ there is one and only one $g \in G$ such that $g.x = y$. 
A regular action is also faithful.
\begin{example}
\textup{
\begin{itemize}
\item A group $G$ acts on itself by (left) translation: $g.x := gx$ for $(g,x) \in G^2$ ($G$ is here written multiplicatively). This action is regular;
\item A subgroup $H$ of a group $G$ also acts on $G$ by translation: $h.x := hx$ for $(h,x) \in H\times G$. This action is faithful 
and if $H$ is a proper subgroup, then the action is not regular;
\item The multiplicative group ${\bbbk}^*$ of a field ${\bbbk}$ acts on ${\bbbk}$ by the multiplication law of 
the group. This action is faithful but not regular since $0$ is fixed by every elements of ${\bbbk}^*$. More generally the action of ${\bbbk}^*$ on a 
${\bbbk}$-vector space by scalar multiplication is also a faithful action 
(in this case the null vector is fixed by any scalar multiplication). 
\end{itemize}
}
\end{example}
\subsection{Group action perfect nonlinearity}
Let $X$ and $Y$ be two finite nonempty sets. A function $f$ is called {\it balanced} if for each $y\in Y$,
\begin{equation}
|\{x \in X | f(x)=y\}|=\frac{|X|}{|Y|}\ .
\end{equation}
With the concept of group actions we now have all the ingredients to recall the notion of group action 
perfect nonlinearity (see \cite{PH05}). 
\begin{definition}
\textup{Let $G$ be a finite group that acts faithfully on a finite nonempty set $X$. Let $H$ be a finite group (written additively). 
A function 
$f : X \rightarrow H$ is called {\it perfect nonlinear} (by respect to the action of $G$ on $X$) or 
{\it $G$-perfect nonlinear} if for each $\alpha \in G^*$, the 
{\it derivative of $f$ in direction $\alpha$}
\begin{equation}
\begin{array}{llll}
d_{\alpha}f: & X & \rightarrow & H\\
& x & \mapsto & f(\alpha.x)-f(x)
\end{array}
\end{equation}
is balanced or in other words for each $\alpha \in G^*$ and each $\beta \in H$,
\begin{equation}
|\{x \in X | d_{\alpha}f(x) = \beta\}|=\frac{|X|}{|H|}\ .
\end{equation}}
\end{definition}
As we can see our definition coincides with the classical one 
(see \cite{CD04}) in the classical 
situations ($G$ acts on itself by left translation).

\section{Doubly perfect nonlinear Boolean permutations}
In the finite fields settings there are two main natural group actions, namely additive and 
multiplicative translations. The first one is the standard used as plaintext and key combination 
process and has been widely studied in terms of (classical) perfect nonlinearity and/or 
bentness. In this contribution we focus on the second one: we construct perfect nonlinear functions 
by respect to multiplication rather than addition called 
{\it multiplicatively perfect nonlinear functions}. 
Moreover in very particular cases, multiplication 
can be seen as an addition of a new finite field. In this paper we exhibit some perfect nonlinear 
functions by respect to both original and new multiplications called {\it doubly perfect nonlinear 
functions}.

\subsection{Multiplicatively perfect nonlinear functions}

Let us begin with a lemma whose proof is a triviality.
\begin{lemma}\label{lemme_pour_egalite_entre_ImInv_et_Ker}
Let $G$ and $H$ be two finite groups (written multiplicatively). 
Let $\lambda$ be a group homomorphism 
from $G$ to $H$. For each $\beta \in \lambda(G)$, 
\begin{equation}
|\lambda^{-1}(\{\beta\})|=|\ker \lambda|\ .
\end{equation}
\end{lemma}

Let $d$ and $m$ be two nonzero integers. We denote by $V(p,m,d)$ 
any $d$ dimensional vector space over the finite field ${\mathsf{GF}}(p^m)$. 
We use the same symbols "$+$" (resp. "$-$") 
to denote both additions (resp. substractions) of $V(p,m,d)$ and 
${\mathsf{GF}}(p^m)$ and $\alpha.v$ is the scalar multiplication of 
$v \in V(p,m,d)$ by $\alpha \in {\mathsf{GF}}(p^m)$. 
\begin{lemma}\label{lemme_sur_les_cardinaux}
Let $d,\ e,\ m,\ n > 0$ be any integers. 
Let $\lambda$ be a group 
homomorphism from $(V(p,m,d),+)$ to $(V(p,n,e),+)$. Let $G$ be a 
subgroup of the group ${\mathsf{GF}}(p^m)^*$. Then for each $\beta \in 
\lambda(V(p,m,d))$ and for each $\alpha \in G^*$, 
\begin{equation}
|\{v \in V(p,m,d) | d_{\alpha}\lambda(v)=\beta\}| = |\lambda^{-1}(\{\beta\})| 
= |\ker \lambda|\ .
\end{equation}
\end{lemma} 
The proof of the previous lemma is not difficult and thus is not given here.
\begin{theorem}\label{theorem_general_mult_PN}
Let $d,\ e,\ m,\ n > 0$ be any integers such that $d^m \geq e^n$. Let $\lambda$ be a group 
epimorphism\footnote{A {\bf{group epimorphism}} is a group 
homomorphism which is onto.} from $(V(p,m,d),+)$ onto $(V(p,n,e),+)$. Then $\lambda$ 
is ${\mathsf{GF}}(p^m)^*$-perfect nonlinear.
\end{theorem}
\begin{proof}
Since $\lambda$ is onto, every $\beta \in V(p,n,e)$ belong to $\lambda(V(p,m,d))$. 
According to lemma~\ref{lemme_sur_les_cardinaux} with $G={\mathsf{GF}}(p^m)^*$, for each 
$\beta \in V(p,n,e)$ 
and for each $\alpha \in {\mathsf{GF}}(p^m)^{**}={\mathsf{GF}}(p^m)\setminus\{0,1\}$, 
$|\{v \in V(p,m,d) | d_{\alpha}\lambda (v)=\beta\}|=|\lambda^{-1}(\{\beta\})|
=|\ker \lambda|$. But $\{\lambda^{-1}(\{\beta\})\}_{\beta \in V(p,n,e)}$ is a 
partition of $V(p,m,d)$. Therefore we have $|V(p,m,d)|=\displaystyle \sum_{\beta \in V(p,n,e)}
|\lambda^{-1}(\{\beta\})|=|\ker \lambda||V(p,n,e)|$. So $|\ker \lambda|=
\frac{|V(p,m,d)|}{|V(p,n,e)|}=p^{md-ne}$.\qed
\end{proof}
In classical situations it is well-known that if a function $f : V(2,m,d) \rightarrow 
V(2,n,e)$ is bent then $md$ is an even integer and $md \geq 2 ne$. Replacing addition by multiplication 
allows us to find "bent" function even if $md$ is an odd integer and/or $2ne > md \geq ne$. When 
$md=ne$ (and $p=2$), 
almost perfect nonlinear (APN) functions are relevant for cryptographic purposes. They are 
defined (see~\cite{NK93}) by the fact that the equation $d_{\alpha}f(x)=\beta$ with $x$ as an unknown 
has at most two solutions 
for each $\alpha \not = 0$ and each $\beta$. The only known examples of APN permutations need $md$ to be 
an odd integer. In our case  
by construction any ${\mathsf{GF}}(p^m)$-linear isomorphism of $V(p,m,d)$ is a 
${\mathsf{GF}}(p^m)^*$-perfect nonlinear; so 
it is also the case for $p=2$ and $md$ an even integer. 

\subsection{Doubly perfect nonlinear Boolean permutations}
The group of units ${\mathsf{GF}}(p^m)^*$ of the finite field 
${\mathsf{GF}}(p^m)$ can be equipped with another multiplication that turns it 
into a unitary commutative ring. Indeed let $\gamma$ be a primitive root of 
${\mathsf{GF}}(p^m)$. The {\it exponential}
\begin{equation}
\begin{array}{llll}
e_{\gamma} : & ({\bbbz}_{p^m -1},+) & \rightarrow & {\mathsf{GF}}(p^m)^*\\
& i & \mapsto & \gamma^i
\end{array}
\end{equation} 
is a group isomorphism (in the remainder we always suppose that such a 
primitive root $\gamma$ is fixed). 
We can use it to turn ${\mathsf{GF}}(p^m)^*$ into a 
commutative unitary ring, isomorphic to the ring of modulo $p^m -1$ integers, 
by\footnote{More rigorously $\gamma^i \times \gamma^j = 
e_{\gamma}(e_{\gamma}^{-1}(\gamma^i)e_{\gamma}^{-1}(\gamma^j)) = 
e_{\gamma}(ij)$. In fact any calculation in the exponent should be understood 
modulo $p^m -1$.} 
$\gamma^{i}\times\gamma^j = \gamma^{ij}$. We call such a  structure 
$({\mathsf{GF}}(p^m),+,0,.,1,\times,\gamma)$ a 
{\it characteristic} $(p,p^m -1)$ {\it field-ring} 
(which means that $({\mathsf{GF}}(p^m),+,0,.,1)$ is a characteristic $2$ field and 
$({\mathsf{GF}}(p^m)^*,.,1,\times,\gamma)$ is a characteristic $p^m -1$ ring 
{\it i.e.} $\gamma^{p^m -1} = 1$, $\gamma^i \not = 1$ for all 
$0 < i < p^m -1$) or {\it double-field} when $({\mathsf{GF}}(p^m)^*,.,1,\times,\gamma)$ is also a field. 
The multiplicative identity of the ring 
$({\mathsf{GF}}(p^m)^*,.,1,\times,\gamma)$ 
is $\gamma^1=\gamma$ and the classical rules of distributivity, absorption and associativity take the 
following forms $\gamma^i \times (\gamma^j \gamma^k)=
(\gamma^i \times \gamma^j)(\gamma^i \times \gamma^k)$, $1\times \gamma^i = 1$, 
$\gamma^i \times (\gamma^j \times \gamma^k)=(\gamma^i \times \gamma^j)\times 
\gamma^k$. The group of units of this ring, $U({\mathsf{GF}}(p^m)^*)$, 
is equal to $\{\gamma^i | i \in U({\bbbz}_{p^m -1})\} = 
\{\gamma^i | (i,p^m -1)=1\}$ (where $(i,j)$ is the greatest common divisor of 
$i$ and $j$) and if $\gamma^i$ is invertible with respect to $\times$ ({\it i.e.} $\gamma^i$ is a unit), 
$(\gamma^{i})^{-1}=\gamma^{\frac{1}{i}}$. 
If $i \not = 0$ is not congruent with 
$1$ modulo $p^m -1$, then it is a zero divisor in ${\bbbz}_{p^m -1}$: it 
exists $j \in {\bbbz}_{p^m-1}^*$ such that $ij = 0$, therefore 
$\gamma^i$ is itself a zero divisor\footnote{More formally we should say a $\times$-divisor 
of $1$.} in ${\mathsf{GF}}(p^m)^*$ because 
$\gamma^i \times \gamma^j = \gamma^{ij}=\gamma^0=1$. This ring is an integral domain if 
and only if $({\bbbz}_{p^m -1},+,0,.,1)$ is itself an integral domain or 
equivalently a (finite) field. So $({\mathsf{GF}}(p^m)^*,.,1,\times,\gamma)$ is 
a finite field if and only if $p^m -1$ is a prime integer. If $p$ is an odd 
prime number then the only possible choice is $p=3$ and $m=1$ (since 
$3^1 -1=2$) because in the other case $p^m -1 > 2$ and is even. The following 
lemma gives a constraint on $m$ when $p=2$. 

\begin{lemma}\label{lemma_Mersenne}
Let $k \in {\bbbn}^*$, $k > 1$. 
Let $m \in {\bbbn}^*$. If $m$ is not a prime integer then so is $k^m -1$. 
\end{lemma}
\begin{proof} 
Suppose that $m = rs$ where both $r$ and $s$ are integers greater or equal to 
$2$. We will prove that $k^{rs}-1=\displaystyle (k^r -1)\sum_{i=1}^{s} 
k^{r(s-i)}$ by induction on the integer $s$.\\
If $s=2$ then $k^{2r}-1 = (k^r -1)(k^r +1)$.\\
Let $s \in {\bbbn}^*$ such that $s \geq 2$. Suppose that for all 
integer $l$ such that $1 < l \leq s$, 
$k^{rl} - 1=\displaystyle (k^r -1)\sum_{i=1}^{l} k^{r(l-i)}$. 
Let us prove that 
$k^{r(s+1)} - 1=\displaystyle (k^r -1)\sum_{i=1}^{s+1} k^{r(s+1-i)}$. 
We have
\begin{equation}
\begin{array}{lll}
k^{r(s+1)} -1 &=& k^{r(s+1)} - k^r + k^r -1\\
&=& k^{r}(k^{rs} -1)+(k^r -1)\\
&=& k^{r}(k^{r}-1)\displaystyle\sum_{i=1}^{s}k^{r(s-i)} + (k^r -1)\ 
\mbox{(by induction hypothesis)}
\end{array}
\end{equation}
\begin{equation}
\begin{array}{lll}
&=& (k^r -1)\displaystyle \left (\sum_{i=1}^s k^{r(s+1-i)} +1\right )\\
&=&(k^r -1)\displaystyle \sum_{i=1}^{s+1}k^{r(s+1-i)}\ .
\end{array}
\end{equation} 
\qed
\end{proof}
An integer of the form $2^q -1$ where $q$ is a prime number is called a 
{\it Mersenne number}. When a Mersenne number is itself a prime integer, it is 
called a {\it Mersenne prime}\footnote{For instance $3=2^2 -1$, $5=2^3 -1$, $31=2^5 -1$ and 
$127=2^7 -1$ are Mersenne prime numbers.}. So given a Mersenne prime $p=2^q -1$, 
$({\mathsf{GF}}(2^q)^*,.,1,\times,\gamma)$ is isomorphic to the prime field 
$({\mathsf{GF}}(p),+,0,.,1)$ (which is identified with $({\bbbz}_{p},+,0,.,1)$) and $({\mathsf{GF}}(2^q),+,0,.,1,\times,\gamma)$ is a characteristic 
$(2,p)$ double-field 
({\it i.e.} $({\mathsf{GF}}(2^q),+,0,.,1)$ is a characteristic $2$ field and 
$({\mathsf{GF}}(2^q)^*,.,1,\times,\gamma)$ is a characteristic $p$ field).\\

We now characterize the existence of some subgroups of units in rings which will be useful 
in the sequel.
\begin{lemma}\label{opp-un-inversible}
Let $R$ be a non-trivial unitary ring\footnote{$R$ is not reduced to $0$.}. 
Then $-1$ is invertible in $R$.
\end{lemma}
\begin{proof}
It is obvious since $(-1)(-1)=1$.\qed
\end{proof}
\begin{lemma}\label{lemma_existence_gr_unite_un_peu_special}
Let $n > 1$. 
The group of units $U({\bbbz}_n)$ contains at least one  subgroup $G$ such that for 
every $i \in G^*$ ({\it i.e.} $i \not = 1$ and $i \in G$), $i-1 \in U({\bbbz}_n)$ if and 
only if $n$ is equal to $2$ or is an odd integer. 
\end{lemma}
\begin{proof}
If $n=2$ then $G=U({\bbbz}_2)=\{1\}$ is a group with the good properties. 
Let suppose that $n > 2$ is an even integer. Then $i$ belongs to 
$U({\bbbz}_n)$ if and only if $(i,n)=1$. Therefore $i$ is an odd integer. 
Then $i-1$ is equal to zero or is an even integer and it is 
 invertible in none of the two cases. Now let suppose that $n$ is an odd integer. Then $2$ is invertible modulo $n$. Since according 
to lemma~\ref{opp-un-inversible} 
(since $n > 1$, ${\bbbz}_n$ is non-trivial), 
$-1$ is a unit, $-2=2(-1)= -1 -1$ is also invertible. The group $G=
\langle -1 \rangle = \{\pm 1\}$ satisfies the assumptions of the lemma. \qed
\end{proof}
We should note that in the particular case where $n$ is a prime number $p$, 
${\bbbz}_p^*=U({\bbbz}_p)$ is such a group $G$. If $n=2^m -1$ then $n$ 
is odd so there is at least one subgroup $G$ of ${\bbbz}_{2^m-1}$ such 
that $\forall i \in G^*$, $i-1 \in U({\bbbz}_{2^m -1})$. If $p$ is an odd 
prime then $p^m -1$ is an even number. So unless the trivial 
case $p=3$ and $m=1$, 
$U({\bbbz}_{p^m -1})$ does not contain any such group $G$.
\begin{lemma}\label{mult_trans_inv}
Let $\gamma^i \in U(({\mathsf{GF}}(p^m)^*,.,1,\times,\gamma))$. Then the map 
\begin{equation}
\begin{array}{llll}
\lambda^{\times}_{\gamma^i} : & {\mathsf{GF}}(p^m)^* & \rightarrow & 
{\mathsf{GF}}(p^m)^*\\
& \gamma^j &\mapsto & \gamma^i \times \gamma^j\ .
\end{array}
\end{equation} 
is a group automorphism of $({\mathsf{GF}}(p^m)^*,.,1)$.
\end{lemma}
\begin{proof}
Since $\times$ is distributive on $.$, $\lambda^{\times}_{\gamma^i}$ is a 
group endomorphism of $({\mathsf{GF}}(p^m)^*,.,1))$. Let $\gamma^j$ such that 
$\gamma^{ij}=1$. This is equivalent to $ij=0$. But $\gamma^i \in 
U({\mathsf{GF}}(p^m)^*)$ so 
$i \in U({\bbbz}_{p^m -1})$ and then $ij=0$ if and 
only if $j=0$. So $\gamma^j = \gamma^0=1$ and $\lambda^{\times}_{\gamma^i}$ is one-to-one also is onto. It is thus an element of 
$\mathit{Aut}(({\mathsf{GF}}(p^m)^*,.,1))$.\qed
\end{proof}

\begin{lemma}\label{lemma_sur_action_sous_grp_units}
Let $G$ be a subgroup of $(U({\mathsf{GF}}(p^m)^*),\times,\gamma)$. Then $G$ 
acts faithfully (by group automorphism) on $({\mathsf{GF}}(p^m)^*,.,1)$ by 
$\rho(\gamma^i) : \gamma^j \mapsto \gamma^i \times \gamma^j$.
\end{lemma}
\begin{proof}
We define    
\begin{equation}
\begin{array}{llll}
\rho : & G & \rightarrow & \mathit{Aut}(({\mathsf{GF}}(p^m)^*,.,1))\\
& \gamma^i & \mapsto & \lambda^{\times}_{\gamma^i}: 
(\gamma^j \mapsto \gamma^i \times \gamma^j)\ .
\end{array}
\end{equation}
(By lemma~\ref{mult_trans_inv} we already know that for each $\gamma^i \in G$, we have 
$\rho(\gamma^i)=\lambda^{\times}_{\gamma^i} \in \mathit{Aut}(({\mathsf{GF}}(p^m)^*,.,1))$.) 
Let's prove that is a group action on ${\mathsf{GF}}(p^m)^*$. 
Let $\gamma^i$ and $\gamma^j$ be elements of 
$G$. Let $\gamma^k \in {\mathsf{GF}}(p^m)^*$. $\rho(\gamma^i \times \gamma^j)
(\gamma^k)=\rho(\gamma^{ij})(\gamma^k)=\gamma^{ij}\times\gamma^k = 
\gamma^{ijk}=\gamma^{i}\times (\gamma^{j} \times \gamma^k)=
(\rho(\gamma^i) \circ \rho(\gamma^j))(\gamma^k)$. Then $\rho$ is a 
group homomorphism from $G$ to $\mathit{Aut}({\mathsf{GF}}(p^m)^*,.,1))$. 
Finally let $\gamma^i \in G$ such that $\rho(\gamma^i)=
\mathit{Id}_{{\mathsf{GF}}(p^m)^*}$. For any $k \in {\bbbz}_{p^m -1}$, 
$\gamma^{ik} = \gamma^k$. So $ik = k$ and in particular $i1 = 1$, therefore 
$i=1$ and $\gamma^i = \gamma^1 = \gamma$. We deduce that $\rho$ is 
one-to-one
and the action is thus faithful. \qed
\end{proof}
\begin{definition}
\textup{
Let $G$ be a group and $X$ be any (nonempty) set. The restriction to $G^*$ 
of a map $f : G \rightarrow X$ is denoted $f^*$.} 
\end{definition}
\begin{theorem}\label{theorem_avec_ssgroupes}
Let $m \in {\bbbn}^*$ such that $m > 1$. Let 
$G$ be a subgroup of $U({\bbbz}_{2^{m} -1})$ such that for each $i \in G^*$, $i-1 \in U({\bbbz}_{2^{m}-1})$ (such a group exists according to 
lemma~\ref{lemma_existence_gr_unite_un_peu_special} since $2^m -1 > 1$ by 
assumption and is an odd number). 
Let $\lambda$ be a field automorphism 
from ${\mathsf{GF}}(2^{m})$ to itself. Then we have
\begin{enumerate}
\item $\lambda$ is $({\mathsf{GF}}(2^{m})^*,.,1)$-perfect nonlinear from 
${\mathsf{GF}}(2^{m})$ to ${\mathsf{GF}}(2^m)$;
\item $\lambda^*$ is $(\gamma^{G},\times,\gamma)$-perfect nonlinear from 
${\mathsf{GF}}(2^{m})^*$ to ${\mathsf{GF}}(2^m)^*$ where $\gamma^G=e_{\gamma}(G)$.
\end{enumerate}
\end{theorem}
\begin{proof}
\begin{enumerate}
\item This result is clear by applying theorem~\ref{theorem_general_mult_PN} 
with ${\mathsf{GF}}(2^m)$ considered as a one-dimensional vectors space over 
itself;
\item Since $\gamma^G = e_{\gamma}(G)$, $\gamma^G$ is a subgroup of 
the group of units of ${\mathsf{GF}}(2^{m})^*$. 
By lemma~\ref{lemma_sur_action_sous_grp_units}, 
$\gamma^G$ acts faithfully on ${\mathsf{GF}}(2^{m})^*$ 
by group automorphism. Because $\lambda$ is a field homomorphism, $\lambda({\mathsf{GF}}(2^{m})^*) \subseteq 
{\mathsf{GF}}(2^m)^*$ and therefore $\lambda^* : 
{\mathsf{GF}}(2^{m})^* \rightarrow {\mathsf{GF}}(2^m)^*$ is a group homomorphism. 
Moreover $\lambda^*$ is onto. Indeed for $y \in {\mathsf{GF}}(2^m)^*$ there is 
$x \in {\mathsf{GF}}(2^{m})$ such that $\lambda (x)=y$. Since $y \not = 0$, 
$x \not = 0$ and therefore $\lambda^{*}(x)=y$. So $\lambda^*$ is a group 
epimorphism (and then a group automorphism). 
Let $\beta \in {\mathsf{GF}}(2^m)^* = \lambda({\mathsf{GF}}(2^{m})^*)$. 
 Let $\gamma^i \in (\gamma^{G})^*$ (so $i \not = 1$). Let's prove that 
$\{\gamma^j \in {\mathsf{GF}}(2^{m})^*| d_{\gamma^i}\lambda^* (\gamma^j)=\beta\} = \gamma^{\frac{1}{j}}\times \lambda^{-1}(\{\beta\})$. We have
\begin{equation}
\begin{array}{llll}
& d_{\gamma^i}\lambda^* (\gamma^j)&=&\beta\\
\Leftrightarrow &\displaystyle \frac{\lambda^{*}(\gamma^i \times \gamma^j)}
{\lambda^* (\gamma^j)}&=&\beta\\
\Leftrightarrow& \displaystyle \lambda (\frac{(\gamma^i \times \gamma^j)}
{\gamma^{j}}) &=& \beta\ \mbox{(because $\lambda$ is a field homomorphism)}\\
\Leftrightarrow & \lambda((\gamma^i \times \gamma^j)(\gamma^{-j}))&=&\beta\\
\Leftrightarrow & \lambda((\gamma^i \times \gamma^j)(\gamma^{-1}\times \gamma^j)) &=&\beta\\
\Leftrightarrow & \lambda((\gamma^i \gamma^{-1})\times \gamma^j)&=&\beta\ 
\mbox{(by distributivity)}\\
\Leftrightarrow & \lambda(\gamma^{i-1}\times \gamma^j)&=&\beta\\
\Leftrightarrow & \gamma^{i-1}\times \gamma^j &\in &\lambda^{-1}(\{\beta\})\ .
\end{array}
\end{equation}
Since $\gamma^i \in (\gamma^G)^*\ \Leftrightarrow\ i \in G^*$ and 
by assumption 
on $G$, $i-1$ is invertible modulo $2^{m}-1$. Then $\gamma^{i-1} \in 
U({\mathsf{GF}}(2^{m})^*)$. According to lemma~\ref{mult_trans_inv}, 
$\lambda_{\gamma^{i-1}}^{\times} \in \mathit{Aut}(({\mathsf{GF}}(2^m)^*,.,1))$. 
Therefore $\gamma^{i-1}\times \gamma^j \in \lambda^{-1}(\{\beta\})$ 
$\Leftrightarrow$ $\gamma^j \in (\lambda^{\times}_{\gamma^{i-1}})^{-1}
\left (\lambda^{-1}(\{\beta\})\right )=\gamma^{\frac{1}{i-1}}\times
\lambda^{-1}(\{\beta\})$. Since $\lambda^{\times}_{\gamma^{\frac{1}{i-1}}}$ is 
a permutation we have $|\lambda^{-1}(\{\beta\})| = 
|\gamma^{\frac{1}{i-1}}\times\lambda^{-1}(\{\beta\})|$. Because 
$\beta \in {\mathsf{GF}}(p^m)^*$, we have $\lambda^{-1}(\{\beta\}) = 
(\lambda^*)^{-1}(\{\beta\})$ and by 
lemma~\ref{lemme_pour_egalite_entre_ImInv_et_Ker}, we deduce that 
$|\gamma^{\frac{1}{i-1}}\times\lambda^{-1}(\{\beta\})|=
|(\lambda^{*})^{-1}(\{\beta\})|=|\ker \lambda^*|$ with 
$\ker \lambda^* = \{x \in {\mathsf{GF}}(2^{m})^*| \lambda^*(x)=1\}=\{x \in {\mathsf{GF}}(2^m)^*|\lambda(x)
=1\}$. In addition 
$\{\lambda^{-1}(\{\beta\})\}_{\beta \in {\mathsf{GF}}(p^{m})^*}$ is a partition 
of ${\mathsf{GF}}(2^{m})^*$. Therefore we have 
\begin{equation}|{\mathsf{GF}}(2^{m})^*|=
\displaystyle \sum_{\beta \in {\mathsf{GF}}(2^m)^*}|\lambda^{-1}(\beta)|=
|\ker \lambda^*||{\mathsf{GF}}(2^m)^*|\ .
\end{equation}
Then for each 
$\gamma^i \in (\gamma^{G})^*$ (or equivalently for each $i \in G^*$) 
and for each $\beta \in {\mathsf{GF}}(2^m)^*$, 
$|\{\gamma^j \in {\mathsf{GF}}(2^{m})^*| 
d_{\gamma^i}\lambda^* (\gamma^j)=\beta\}|=|\ker \lambda^*|=1$.
\end{enumerate}
\qed
\end{proof}
\begin{definition}
\textup{
Let $p=2^q -1$ be a Mersenne prime number. A function $f : {\mathsf{GF}}(2^q)
\rightarrow {\mathsf{GF}}(2^q)$ such that $f(\alpha)\not=0$ for all 
invertible $\alpha \in {\mathsf{GF}}(2^q)$ 
is called {\it doubly perfect nonlinear} if 
\begin{enumerate}
\item $f$ is $({\mathsf{GF}}(2^q)^*,.,1)$-perfect nonlinear from 
${\mathsf{GF}}(2^q)$ to itself;
\item $f^*$ is $({\mathsf{GF}}(2^q)^{**},\times,\gamma)$-perfect nonlinear 
from ${\mathsf{GF}}(2^q)^{*}$ to itself.
\end{enumerate}
}
\end{definition}
Since the group of field automorphisms of a finite field ${\mathsf{GF}}(p^m)$ 
is identical to the Galois group of the degree $m$ extension ${\mathsf{GF}}(p^m)$ over its prime field which is a cyclic group generated by the 
Frobenius automorphism 
\begin{equation}
\begin{array}{llll}
{\mathcal{F}}_p:&{\mathsf{GF}}(p^m)&\rightarrow&{\mathsf{GF}}(p^m)\\
& x & \mapsto & x^p
\end{array}
\end{equation} 
every field automorphism $\lambda$ can be written as 
${\mathcal{F}}_p^r$ for one $r$ such that $0\leq r \leq m-1$.
We now give a nice result that asserts the existence 
of a Boolean permutation over ${\mathsf{GF}}(2^q)$, where $p=2^q -1$ is a Mersenne prime, which is merely both 
$({\mathsf{GF}}(2^q)^*,.,1)$ and $({\mathsf{GF}}(2^q)^{**},\times,\gamma)$-perfect 
nonlinear {\it i.e.} doubly perfect nonlinear.
\begin{theorem}\label{thm_frobenius}
Let $p=2^q -1$ be a Mersenne prime number. Let $\lambda = 
{\mathcal{F}}_2^r$ (for any $0 \leq r \leq q-1$) 
be a field automorphism of ${\mathsf{GF}}(2^q)$. Then $\lambda$ is a doubly perfect nonlinear permutation.
\end{theorem}
\begin{proof}
Because $p=2^q -1$ is a prime number, ${\mathsf{GF}}(2^q)^*$ is isomorphic 
to the 
field ${\mathsf{GF}}(p)={\bbbz}_p$. Therefore we can choose 
$G={\bbbz}_p^*$ as a group such that for each $i \in G^*$, $i-1$ is invertible modulo $p$. 
Then $\gamma^{G}=U({\mathsf{GF}}(2^q)^*)={\mathsf{GF}}(2^q)^{**}=
{\mathsf{GF}}(2^q)\setminus\{0,1\}$. According to 
theorem~\ref{theorem_avec_ssgroupes}, $\lambda$ is 
$({\mathsf{GF}}(2^q)^*,.,1)$-perfect nonlinear and 
$\lambda^*$ is $({\mathsf{GF}}(2^q)^{**},\times,\gamma)$-perfect nonlinear.\qed
\end{proof}


\end{document}